\documentclass[runningheads]{llncs}
\usepackage{amsmath, amsfonts, amssymb}
\usepackage{xcolor}
\usepackage{comment}
\usepackage{graphicx,caption}
\usepackage{tikz}
\usepackage{algorithm}
\usepackage[noend]{algpseudocode}
\usepackage{todonotes}
\usepackage{tikz, subcaption, ifthen}

\newcommand{\OPT}{\mathrm{OPT}}
\renewcommand{\P}{\mathrm{P}}
\newcommand{\NP}{\mathrm{NP}}
\newcommand{\DTIME}{\mathrm{DTIME}}

\usepackage{graphicx}
\usepackage[noadjust]{cite}
\begin{document}
\title{Approximation Algorithms for Priority Steiner Tree Problems}
%
%
\author{Faryad Darabi Sahneh\inst{1} \and
Stephen Kobourov\inst{1}\and
Richard Spence\inst{1}}
\authorrunning{F. Darabi Sahneh et al.}
%
\institute{University of Arizona, Tucson, AZ 85721, USA}
\maketitle              
\begin{abstract}
In the Priority Steiner Tree (PST) problem, we are given an undirected graph $G=(V,E)$ with a source $s \in V$ and terminals $T \subseteq V \setminus \{s\}$, where each terminal $v \in T$ requires a nonnegative priority $P(v)$. The goal is to compute a minimum weight Steiner tree containing edges of varying rates such that the path from $s$ to each terminal $v$ consists of edges of rate greater than or equal to $P(v)$. The PST problem with $k$ priorities admits a $\min\{2 \ln |T| + 2, k\rho\}$-approximation [Charikar et al., 2004], and is hard to approximate with ratio $c \log \log n$ for some constant $c$ [Chuzhoy et al., 2008]. In this paper, we first strengthen the analysis provided by [Charikar et al., 2004] for the $(2 \ln |T| + 2)$-approximation to show an approximation ratio of $\lceil \log_2 |T| \rceil + 1 \le 1.443 \ln |T| + 2$, then provide a very simple, parallelizable algorithm which achieves the same approximation ratio. We then consider a more difficult node-weighted version of the PST problem, and provide a $2 \ln (|T|+1)$-approximation using extensions of the spider decomposition by [Klein \& Ravi, 1995]. This is the first result for the PST problem in node-weighted graphs. Moreover, the approximation ratios for all above algorithms are tight.

\keywords{
priority Steiner tree \and 
approximation algorithms \and network design}
\end{abstract}

\section{Introduction}
We consider generalizations of the Steiner tree and node-weighted Steiner tree (NWST) problems in graphs where the terminals $T$ possess varying priority or quality of service (QoS) requirements, in which we seek to connect the terminals using edges of the appropriate rate or better. These problems have applications in multimedia and electric power distribution~\cite{turletti1994,maxemchuk1997,Balakrishnan1994}, multi-level graph visualization~\cite{MLST2018}, and other network design problems where a source or root is to be connected to a set of heterogeneous receivers possessing different bandwidth or priority requests. We define a Priority Steiner Tree (PST) as follows:

\begin{definition}[Priority Steiner Tree (PST)]\label{def:PST}
Given an undirected graph $G=(V,E)$, a source $s \in V$, and terminals $T \subseteq V \setminus \{s\}$, where each terminal $v \in T$ requires a nonnegative priority $P(v)$, a \emph{PST} is a tree $\mathcal{T}\subseteq G$ rooted at $s$ containing edges of varying rates such that for all terminals $v \in T$, the $s$--$v$ path in $\mathcal{T}$ consists of edges of rate $P(v)$ or higher.
\end{definition} 
We denote by $k$ the number of distinct priorities. Vertices in $V \setminus (T \cup \{s\})$ have zero priority but may be included in $\mathcal{T}$. Let $w(e,r)$ denote the weight of edge $e$ at rate $r$. We assume $w(e,0) = 0$ and $w(e,r_1) \le w(e,r_2)$ for all $0 \le r_1 \le r_2$ and edges $e$ (i.e., higher-rate edges weigh at least as much as lower-rate edges). The weight of a PST $\mathcal{T}$ is the sum of the weights of the edges in $\mathcal{T}$ at their respective rates, namely $w(\mathcal{T}) := \sum_{e \in E(\mathcal{T})} w(e, R(e))$.

\begin{problem}[\textsc{Priority Steiner Tree} problem]\label{problem:PST}
Given a graph $G=(V,E)$, source $s$, terminals $T \subseteq V$, priorities $P(\cdot)$, and edge weights $w:E \times \mathbb{R}_{\ge 0} \to \mathbb{R}_{\ge 0}$, compute a PST $\mathcal{T}$ with minimum weight.
\end{problem}

 While Problem~\ref{problem:PST} in the case where edge weights are proportional to rate (i.e., $w(e,r) = r \cdot w(e,1)$ for all $e \in E$ and $r \ge 0$) admits $O(1)$--approximations~\cite{Charikar2004,Karpinski2005,MLST2018}, the best known approximation ratio for \textsc{Priority Steiner tree} with arbitrary weights is $\min\{2\ln |T| + 2, k\rho\}$ by Charikar et al.~\cite{Charikar2004} (see Section~\ref{section:PST}). On the other hand, Chuzhoy et al.~\cite{Chuzhoy2008} show that $\textsc{Priority Steiner tree}$ cannot be approximated with ratio $c \log \log n$ for some constant $c$ unless $\NP \subseteq \DTIME(n^{O(\log \log \log n)})$, even with unit edge weights\footnote{We remark that the formulation of \textsc{Priority Steiner Tree} given in~\cite{Chuzhoy2008} is slightly more specific; each edge has a single weight $c_e$ as well as a quality of service (priority) $Q(e)$ on input, and the goal is to compute a Steiner tree such that the path from root to each terminal $v$ uses edges of quality of service greater than or equal to $P(v)$.}.

In Section~\ref{section:greedyPNWST}, we introduce a node-weighted variant of \textsc{Priority Steiner Tree}, called \textsc{Priority NWST} (Def.~\ref{def:PNWST}). Here we assume edges have zero weight, as an instance with edge and vertex weights can be converted to an instance with only vertex weights by subdividing each edge $uv$ into two edges $uw$, $wv$ and assigning the weight of edge $uv$ to vertex $w$.

\begin{definition}[Priority Node-Weighted Steiner Tree (PNWST)] \label{def:PNWST} Given an undirected graph $G=(V,E)$, source $s$, and terminals $T \subseteq V \setminus \{s\}$, where each terminal $v \in T$ requires a nonnegative priority $P(v)$, a \emph{priority node-weighted Steiner tree} (PNWST) is a tree $\mathcal{T}$ rooted at $s$ containing \underline{vertices} of varying rates $R(v)$ such that for all terminals $v \in T$, the $s$--$v$ path in $\mathcal{T}$ consists of vertices of rate $P(v)$ or higher.
\end{definition}
In particular, we require 
$R(v) \ge P(v)$ for all $v \in T$. Further, we can assume w.l.o.g. that the path from $s$ to each terminal uses vertices of non-increasing rate (see Def.~\ref{def:rate-tree}). 
As in the NWST problem, it is conventional to 
also assume terminals have zero weight, as they must be included in any feasible solution; thus, we assume $w(v,r) = 0$ for $0 \le r \le P(v)$ and $w(v,r_1) \le w(v,r_2)$ for all $0 \le r_1 \le r_2$.
The weight of a PNWST $\mathcal{T}$ with \emph{vertex} rates $R(\cdot)$ 
is $w(\mathcal{T}) := \sum_{v \in V(\mathcal{T})} w(v, R(v))$.
\begin{problem}[\textsc{Priority NWST} problem]\label{problem:PNWST}
Given a graph $G=(V,E)$, source $s$, terminals $T \subseteq V \setminus \{s\}$, vertex priorities $P(\cdot)$, and vertex weights $w:V \times \mathbb{R}_{\ge 0} \to \mathbb{R}_{\ge 0}$, compute a PNWST $\mathcal{T}$ with minimum weight.
\end{problem}

The $\textsc{Priority NWST}$ problem generalizes the NWST problem, and hence cannot be approximated with ratio $(1-o(1))\ln |T|$ unless $\P = \NP$~\cite{feige1998threshold,KLEIN1995104,dinur2013}, via a reduction from the set cover problem. In Section~\ref{section:greedyPNWST}, we show that the \textsc{Priority NWST} problem admits a $2 \ln (|T|+1)$--approximation (Theorem~\ref{thm:greedyPNWST}) using extensions of the spider decomposition given by Klein and Ravi~\cite{KLEIN1995104} to accommodate the priority constraints of the \textsc{Priority NWST} problem. The generalization is not immediately obvious; in particular it is not immediate whether an instance of \textsc{Priority NWST} can be formulated as an instance of NWST. However, NWST and \textsc{Priority NWST} can be easily reduced to Steiner arborescence (or directed Steiner tree), which admits a quasi-polynomial $O\left(\frac{\log^2 |T|}{\log \log |T|}\right)$-approximation~\cite{grandoni2019dst}.

\paragraph*{Notation.} A graph $G=(V,E)$ with $n=|V|$ and $m=|E|$ is undirected and connected, unless stated otherwise. Given terminals $u,v \in T$ for the PST problem, denote by $\sigma(u,v)$ the weight of a minimum weight $u$--$v$ path in $G$ using edges of rate $\min\{P(u), P(v)\}$, and let $p_{uv}$ denote such a path. For terminals $u,v \in T$ in the \textsc{Priority NWST} problem, we define $\sigma(u,v)$ to be the weight of a minimum $u$--$v$ path using vertices of rate $\min\{P(u), P(v)\}$ not including the endpoints $u$ and $v$, and similarly define $\sigma_b(u,v)$ to be the weight of a minimum weight vertex-weighted path using vertices of rate $b$, so that $\sigma(u,v) = \sigma_{\min\{P(u), P(v)\}} (u,v)$. In particular, we have $\sigma_b(v,v)=0$. Note that $\sigma$ is symmetric but does not satisfy the triangle inequality, and is not a metric.
Let $\rho$ denote an approximation ratio for the (edge-weighted) Steiner tree problem, and let $\text{STEINER}(n)$ denote the running time of such an approximation algorithm on an $n$-vertex graph. We denote by $\OPT$ the weight of a min-weight PST or PNWST. Lastly, for $n \in \mathbb{Z}^+$, we denote by $[n]$ the set $\{1,2,\ldots,n\}$.

\subsection{Related work}\label{subsection:related}
The Steiner tree problem in graphs has been studied in a wide variety of contexts; see the compendium~\cite{Hauptmann2015}. The (edge-weighted) Steiner tree problem admits a folklore $2\left(1 - \frac{1}{|T|}\right)$--approximation, and is approximable with ratio $\rho = \ln 4 + \varepsilon \approx 1.387$~\cite{Byrka2013}, but NP-hard to approximate with ratio $\frac{96}{95} \approx 1.01$~\cite{Chlebnik2008}. As stated previously, NWST cannot be approximated with ratio $(1-o(1))\ln |T|$ unless $\P = \NP$~\cite{feige1998threshold,KLEIN1995104,dinur2013}, but algorithms with logarithmic approximation ratio exist. Klein and Ravi~\cite{KLEIN1995104} give a $2\ln|T|$--approximation for NWST, which was improved to $1.61 \ln |T|$ and a less practical $(1.35+\varepsilon) \ln |T|$ by Guha and Khuller~\cite{GUHA199957}. Demaine et al.~\cite{demaine2009nwst} give an $O(1)$--approximation for NWST when the input graph $G$ is $H$--minor free, and a 6-approximation when $G$ is planar. Naor et al.~\cite{naor2011online} give a randomized $O(\log n \log^2 |T|)$-approximation algorithm for the online version.

The (edge-weighted) \textsc{Priority Steiner tree} problem and variants thereof have been studied under various other names including Hierarchical Network Design~\cite{current1986}, Multi-Level (or $k$-Level) Network Design~\cite{Balakrishnan1994}, Multi-Tier Tree~\cite{mirchandani1996MTT}, Grade of Service Steiner Tree~\cite{xue2001gosst}, Quality of Service Multicast Tree~\cite{Charikar2004,Karpinski2005}, and Multi-Level Steiner Tree~\cite{MLST2018,ahmed2020kruskalbased}. 
Earlier results on this problem typically consider a small number of priorities or restricted definition of weight~\cite{current1986,Balakrishnan1994}. In the special case where edge weights are proportional to rate, Charikar et al.~\cite{Charikar2004} give the first $O(1)$--approximations with approximation ratios $4\rho$ and $e\rho \approx 4.214$ (with $\rho \approx 1.55$~\cite{robins2000}) independent of the number of priorities $k$. Karpinski et al. \cite{Karpinski2005} give a slightly stronger variant of the $e\rho$–approximation \cite{Charikar2004} which achieves approximation ratio 3.802. Ahmed et al.~\cite{MLST2018} give an approximation ratio of $2.351\rho \approx 3.268$ for $k \le 100$. Xue et al.~\cite{xue2001gosst} consider this problem where the terminals are embedded in the Euclidean plane, and give $\frac{4}{3}\rho$ (resp. $\frac{5+4\sqrt{2}}{7}\rho \approx 1.522\rho$)--approximations for two (resp. three) different priorities. Integer programming formulations have been proposed and evaluated over realistic problem instances~\cite{MLST2018,risso2020qosmt}.

If edge weights are not necessarily proportional to rate, Charikar et al.~\cite{Charikar2004} gave a simple $\min\{2 \ln |T|+2, k\rho\}$-approximation (see Section~\ref{section:PST}), which remains the best known to date. Recently, Ahmed et al.~\cite{ahmed2020kruskalbased} proposed an approximation based on Kruskal's MST algorithm which achieves the same approximation ratio, and provided an experimental study comparing the two methods. Chuzhoy et al.~\cite{Chuzhoy2008} showed that \textsc{Priority Steiner tree} cannot be approximated with ratio $c \log \log n$ for some constant $c$ unless $\NP \subseteq \DTIME(n^{O(\log \log \log n)})$. Angelopoulos~\cite{Angelopoulos2009} showed that every deterministic online algorithm for online \textsc{Priority Steiner tree} has ratio $\Omega(\min\{k \log \frac{|T|}{k}, |T|\})$. Interestingly, no node-weighted variant of \textsc{Priority Steiner tree} has been studied in existing literature. 
However, a related problem is the (single-source) node-weighted buy-at-bulk problem (NSS-BB) studied by Chekuri et al.~\cite{chekuri2010approximation}, who show 
a $3H_{|T|} = O(\log |T|)$--approximation for NSS-BB by giving a randomized algorithm then derandomizing it using an LP relaxation, where $H_n = \frac{1}{1}+\frac{1}{2}+\ldots+\frac{1}{n}$ is the $n^{\text{th}}$ harmonic number.

\subsection{Our results}
In Section~\ref{section:PST}, we strengthen the analysis of the simple $(2 \ln |T| + 2)$-approximation (Algorithm~\ref{alg:qosmt}) by Charikar et al.~\cite{Charikar2004} to show that it is a $\lceil\log_2 |T|\rceil + 1 \le (1.443 \ln |T| + 2)$-approximation. We then give a parallelizable algorithm (Algorithm~\ref{alg:pst}) with the same approximation ratio that does not require that terminals be connected sequentially or in a particular order. This contrasts with the inherently serial Algorithm~\ref{alg:qosmt}~\cite{Charikar2004}, where the shortest path for each terminal depends on the partial PST computed at the previous iteration.

\begin{theorem}
Algorithm~\ref{alg:qosmt}~\cite{Charikar2004} is a $(\lceil \log_2 |T| \rceil + 1)$-approximation for \textsc{Priority Steiner tree} with running time $O(nm + n^2 \log n)$, and there is a parallelizable algorithm for \textsc{Priority Steiner tree} with the same approximation ratio.
\end{theorem}
Moreover, the approximation ratio is tight up to a factor of 2, as there exists an input graph in which Algorithms~\ref{alg:qosmt}--\ref{alg:pst} may output a PST with weight $\frac{1}{2} \log_2 |T| + 1$ times the optimum~\cite{imase91dst}. In Section~\ref{section:greedyPNWST}, we show the following result for  \textsc{Priority NWST}:
\begin{theorem}\label{thm:greedyPNWST}
There exists a $2 \ln (|T|+1)$--approximation algorithm for \textsc{Priority NWST} with running time $O(n^4 k \log n)$.
\end{theorem}
To the best of our knowledge, this is the first approximation algorithm for \textsc{Priority NWST}, and is the main technical contribution of this paper. The analysis extends the spider decomposition of Klein and Ravi~\cite{KLEIN1995104} in their greedy $(2\ln|T|)$--approximation for the NWST problem, to accommodate priority constraints in the \textsc{Priority NWST} problem. Note the additional $+1$ arises as we do not consider the source $s$ a terminal. Moreover, the approximation ratio is tight. 

\section{Priority Steiner Tree: Two logarithmic approximations} \label{section:PST}

We first review the greedy $\min\{2 \ln |T| + 2, k\rho\}$ approximation for \textsc{Priority Steiner tree} given by Charikar et al.~\cite{Charikar2004}. This returns the better solution of two sub-algorithms; we focus primarily on the $(2 \ln |T| + 2)$-approximation (Algorithm~\ref{alg:qosmt}). This algorithm sorts the terminals $T$ from highest to lowest priority. Then for $i=1$, \ldots, $|T|$, the $i^{\text{th}}$ terminal $v_i$ in the sorted list is connected to the existing tree (containing the source $s$) using a minimum weight path of rate $P(v_i)$. The weight of this path is the \emph{connection cost} of $v_i$. Cycles can be removed in the end by removing an edge from each cycle with the lowest rate.

\begin{algorithm}[h!]

\caption{$R(\cdot) = \text{QoSMT}(\text{graph }G, \text{ priorities }P, \text{ edge weights }w, \text{ source } s)$ \cite{Charikar2004}}\label{alg:qosmt}
\begin{algorithmic}[1]
\State Sort terminals $T$ by decreasing priority $P(\cdot)$
\State Initialize $V' = \{s\}$, $R(e) = 0$ for $e \in E$
\For{$i=1,2,\ldots,|T|$}
\State Connect $i^{\text{th}}$ terminal $v_i$ to $V'$ using minimum weight path $p_i$ of rate $P(v_i)$\label{line:qosmt-dijkstra}
\State $R(e) = P(v_i)$ for $e \in p_i$
\State $V' = V' \cup V(p_i)$
\EndFor
\State Remove lowest-rate edge from each cycle
\State \Return{edge rates $R(\cdot)$}
\end{algorithmic}
\end{algorithm}

Algorithm~\ref{alg:qosmt} is based on a $(\log_2 |T|)$-approximation for an online Steiner tree problem analyzed by Imase and Waxman~\cite{imase91dst}; however, Charikar et al.~\cite{Charikar2004} give a simpler analysis which proves a weaker approximation ratio of $2 \ln |T| + 2$, based on the following lemma:
\begin{lemma}[\cite{Charikar2004}]
\label{lemma:qosmt-cc}
For $1 \le x \le |T|$, the $x^{\text{th}}$ most expensive connection cost incurred by Algorithm~\ref{alg:qosmt} is at most $\frac{2\OPT}{x}$.
\end{lemma}
Lemma~\ref{lemma:qosmt-cc} implies the weight of the PST is at most $2\OPT\left(\frac{1}{1} + \frac{1}{2} + \ldots + \frac{1}{|T|}\right) = 2\OPT H_{|T|} \le (2\ln |T| + 2) \OPT$. Line~\ref{line:qosmt-dijkstra} can be executed by running Dijkstra's algorithm from $v_i$ with edge weights $w(\cdot, P(v_i))$ until reaching a vertex in $V'$; hence Algorithm~\ref{alg:qosmt} runs in $O(nm + n^2 \log n)$ time.

We strengthen the analysis by Charikar et al.~\cite{Charikar2004} to prove an approximation ratio of $\lceil \log_2 |T|\rceil + 1$, thus matching the result for the online Steiner tree problem~\cite{imase91dst}. Instead of an upper bound on the $x^{\text{th}}$ most expensive connection cost, we establish a bound on the $\frac{|T|}{2}$ \emph{least} expensive connection costs; a similar technique was used in~\cite{marathe1998bicriteria} for a bicriteria diameter-constrained Steiner tree problem. For simplicity, we assume w.l.o.g. $|T|$ is a power of 2; this can be done by adding up to one dummy terminal of priority 1 to each terminal, connected with a zero-weight edge.

\begin{lemma}\label{lemma:pst-cc}
The sum of the $\frac{|T|}{2}$ least expensive connection costs incurred by Algorithm~\ref{alg:qosmt} is at most $\OPT$.
\end{lemma}
\begin{proof}
Let $\mathcal{T}^*$ be a minimum weight PST rooted at $s$ with weight $\OPT$ and edge rates $R^*(\cdot)$. Consider a depth-first traversal of $\mathcal{T}^*$ starting and ending at $s$ (Fig.~\ref{fig:dfs-traversal-example}).
 
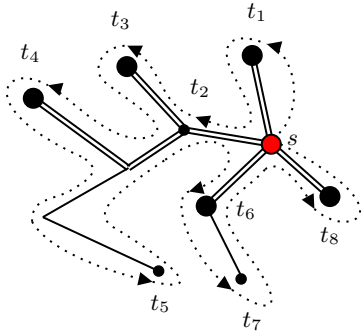
\begin{figure}[h]
\centering
\vspace{-.4cm}\tikzset{every picture/.style={line width=0.75pt}} 

\begin{tikzpicture}[x=0.6pt,y=0.6pt,yscale=-1,xscale=1]

\draw    (371.47,39.69) -- (383.47,95.69)(368.53,40.31) -- (380.53,96.31) ;
\draw [shift={(370,40)}, rotate = 77.91] [color={rgb, 255:red, 0; green, 0; blue, 0 }  ][fill={rgb, 255:red, 0; green, 0; blue, 0 }  ][line width=0.75]      (0, 0) circle [x radius= 3.35, y radius= 3.35]   ;
\draw    (383.03,97.09) -- (342.03,136.09)(380.97,94.91) -- (339.97,133.91) ;
\draw    (383,94.88) -- (420,127.88)(381,97.12) -- (418,130.12) ;
\draw    (381.75,97.48) -- (327.75,88.48)(382.25,94.52) -- (328.25,85.52) ;
\draw    (326.9,88.02) -- (289.9,48.02)(329.1,85.98) -- (292.1,45.98) ;
\draw    (328.83,88.25) -- (292.83,112.25)(327.17,85.75) -- (291.17,109.75) ;
\draw  [fill={rgb, 255:red, 0; green, 0; blue, 0 }  ,fill opacity=1 ] (364,40) .. controls (364,36.69) and (366.69,34) .. (370,34) .. controls (373.31,34) and (376,36.69) .. (376,40) .. controls (376,43.31) and (373.31,46) .. (370,46) .. controls (366.69,46) and (364,43.31) .. (364,40) -- cycle ;
\draw  [fill={rgb, 255:red, 0; green, 0; blue, 0 }  ,fill opacity=1 ] (335,135) .. controls (335,131.69) and (337.69,129) .. (341,129) .. controls (344.31,129) and (347,131.69) .. (347,135) .. controls (347,138.31) and (344.31,141) .. (341,141) .. controls (337.69,141) and (335,138.31) .. (335,135) -- cycle ;
\draw  [fill={rgb, 255:red, 0; green, 0; blue, 0 }  ,fill opacity=1 ] (413,129) .. controls (413,125.69) and (415.69,123) .. (419,123) .. controls (422.31,123) and (425,125.69) .. (425,129) .. controls (425,132.31) and (422.31,135) .. (419,135) .. controls (415.69,135) and (413,132.31) .. (413,129) -- cycle ;
\draw  [fill={rgb, 255:red, 0; green, 0; blue, 0 }  ,fill opacity=1 ] (285,47) .. controls (285,43.69) and (287.69,41) .. (291,41) .. controls (294.31,41) and (297,43.69) .. (297,47) .. controls (297,50.31) and (294.31,53) .. (291,53) .. controls (287.69,53) and (285,50.31) .. (285,47) -- cycle ;
\draw  [fill={rgb, 255:red, 0; green, 0; blue, 0 }  ,fill opacity=1 ] (324,87) .. controls (324,85.34) and (325.34,84) .. (327,84) .. controls (328.66,84) and (330,85.34) .. (330,87) .. controls (330,88.66) and (328.66,90) .. (327,90) .. controls (325.34,90) and (324,88.66) .. (324,87) -- cycle ;
\draw    (232.89,65.79) -- (292.89,109.79)(231.11,68.21) -- (291.11,112.21) ;
\draw  [fill={rgb, 255:red, 0; green, 0; blue, 0 }  ,fill opacity=1 ] (226,67) .. controls (226,63.69) and (228.69,61) .. (232,61) .. controls (235.31,61) and (238,63.69) .. (238,67) .. controls (238,70.31) and (235.31,73) .. (232,73) .. controls (228.69,73) and (226,70.31) .. (226,67) -- cycle ;
\draw    (292,111) -- (238,142) ;
\draw  [fill={rgb, 255:red, 0; green, 0; blue, 0 }  ,fill opacity=1 ] (308,176) .. controls (308,174.34) and (309.34,173) .. (311,173) .. controls (312.66,173) and (314,174.34) .. (314,176) .. controls (314,177.66) and (312.66,179) .. (311,179) .. controls (309.34,179) and (308,177.66) .. (308,176) -- cycle ;
\draw    (341,135) -- (364,181) ;
\draw  [fill={rgb, 255:red, 0; green, 0; blue, 0 }  ,fill opacity=1 ] (360,181) .. controls (360,179.34) and (361.34,178) .. (363,178) .. controls (364.66,178) and (366,179.34) .. (366,181) .. controls (366,182.66) and (364.66,184) .. (363,184) .. controls (361.34,184) and (360,182.66) .. (360,181) -- cycle ;
\draw    (238,142) -- (311,176) ;
\draw  [fill={rgb, 255:red, 255; green, 0; blue, 0 }  ,fill opacity=1 ] (376,96) .. controls (376,92.69) and (378.69,90) .. (382,90) .. controls (385.31,90) and (388,92.69) .. (388,96) .. controls (388,99.31) and (385.31,102) .. (382,102) .. controls (378.69,102) and (376,99.31) .. (376,96) -- cycle ;
\draw  [dash pattern={on 0.84pt off 2.51pt}]  (387,89) .. controls (402.39,73.61) and (397.74,41.25) .. (381.81,33.77) ;
\draw [shift={(379.2,32.8)}, rotate = 375.52] [fill={rgb, 255:red, 0; green, 0; blue, 0 }  ][line width=0.08]  [draw opacity=0] (8.93,-4.29) -- (0,0) -- (8.93,4.29) -- cycle    ;
\draw  [dash pattern={on 0.84pt off 2.51pt}]  (337.29,79.94) .. controls (386.4,97.68) and (367.01,70.4) .. (357.45,50.93) .. controls (347.7,31.05) and (359.2,24.3) .. (379.2,32.8) ;
\draw [shift={(334.2,78.8)}, rotate = 20.67] [fill={rgb, 255:red, 0; green, 0; blue, 0 }  ][line width=0.08]  [draw opacity=0] (8.93,-4.29) -- (0,0) -- (8.93,4.29) -- cycle    ;
\draw  [dash pattern={on 0.84pt off 2.51pt}]  (294.15,35.38) .. controls (313.9,47.19) and (311.28,72.12) .. (334.2,78.8) ;
\draw [shift={(291.2,33.8)}, rotate = 25.56] [fill={rgb, 255:red, 0; green, 0; blue, 0 }  ][line width=0.08]  [draw opacity=0] (8.93,-4.29) -- (0,0) -- (8.93,4.29) -- cycle    ;
\draw  [dash pattern={on 0.84pt off 2.51pt}]  (243.93,60.1) .. controls (255.65,67.14) and (281.12,96.4) .. (303.2,91.8) .. controls (327.2,86.8) and (281.95,62.8) .. (275.58,55.3) .. controls (269.2,47.8) and (275.2,33.8) .. (291.2,33.8) ;
\draw [shift={(241.2,58.8)}, rotate = 16.7] [fill={rgb, 255:red, 0; green, 0; blue, 0 }  ][line width=0.08]  [draw opacity=0] (8.93,-4.29) -- (0,0) -- (8.93,4.29) -- cycle    ;
\draw  [dash pattern={on 0.84pt off 2.51pt}]  (262.2,103.8) .. controls (193.2,72.8) and (213.2,42.8) .. (241.2,58.8) ;
\draw  [dash pattern={on 0.84pt off 2.51pt}]  (231.2,144.8) .. controls (231.2,128.8) and (244.2,132.05) .. (256.2,124.8) .. controls (268.2,117.55) and (276.2,111.8) .. (262.2,103.8) ;
\draw  [dash pattern={on 0.84pt off 2.51pt}]  (231.2,144.8) .. controls (239.98,160.4) and (269.66,174.1) .. (304.51,182.19) ;
\draw [shift={(307.2,182.8)}, rotate = 192.53] [fill={rgb, 255:red, 0; green, 0; blue, 0 }  ][line width=0.08]  [draw opacity=0] (8.93,-4.29) -- (0,0) -- (8.93,4.29) -- cycle    ;
\draw  [dash pattern={on 0.84pt off 2.51pt}]  (307.2,182.8) .. controls (351.2,192.8) and (300.2,158.8) .. (283.2,148.8) .. controls (266.2,138.8) and (287.2,124.8) .. (315.2,105.8) ;
\draw  [dash pattern={on 0.84pt off 2.51pt}]  (315.2,105.8) .. controls (354.4,76.4) and (373.43,111.35) .. (332.76,124.99) ;
\draw [shift={(330.2,125.8)}, rotate = 343.53999999999996] [fill={rgb, 255:red, 0; green, 0; blue, 0 }  ][line width=0.08]  [draw opacity=0] (8.93,-4.29) -- (0,0) -- (8.93,4.29) -- cycle    ;
\draw  [dash pattern={on 0.84pt off 2.51pt}]  (330.2,125.8) .. controls (318.5,128.73) and (328.66,164.92) .. (355.13,186.19) ;
\draw [shift={(357.2,187.8)}, rotate = 216.87] [fill={rgb, 255:red, 0; green, 0; blue, 0 }  ][line width=0.08]  [draw opacity=0] (8.93,-4.29) -- (0,0) -- (8.93,4.29) -- cycle    ;
\draw  [dash pattern={on 0.84pt off 2.51pt}]  (357.2,187.8) .. controls (381.2,204.8) and (394.2,201.8) .. (367.2,165.8) .. controls (340.2,129.8) and (359.2,132.8) .. (378.2,111.8) ;
\draw  [dash pattern={on 0.84pt off 2.51pt}]  (378.2,111.8) .. controls (387.55,104.32) and (400.4,121.32) .. (407.75,134.19) ;
\draw [shift={(409.2,136.8)}, rotate = 241.7] [fill={rgb, 255:red, 0; green, 0; blue, 0 }  ][line width=0.08]  [draw opacity=0] (8.93,-4.29) -- (0,0) -- (8.93,4.29) -- cycle    ;
\draw  [dash pattern={on 0.84pt off 2.51pt}]  (409.2,136.8) .. controls (430.2,156.8) and (471.2,138.8) .. (388,96) ;

\draw (365,5.4) node [anchor=north west][inner sep=0.75pt]    {$t_{1}$};
\draw (329,53.4) node [anchor=north west][inner sep=0.75pt]    {$t_{2}$};
\draw (279,10.4) node [anchor=north west][inner sep=0.75pt]    {$t_{3}$};
\draw (222,31.4) node [anchor=north west][inner sep=0.75pt]    {$t_{4}$};
\draw (304,189.4) node [anchor=north west][inner sep=0.75pt]    {$t_{5}$};
\draw (358,130.4) node [anchor=north west][inner sep=0.75pt]    {$t_{6}$};
\draw (362,200.4) node [anchor=north west][inner sep=0.75pt]    {$t_{7}$};
\draw (411,146.4) node [anchor=north west][inner sep=0.75pt]    {$t_{8}$};
\draw (390,88) node [anchor=north west][inner sep=0.75pt]    {$s$};

\end{tikzpicture}
\caption{Depth-first traversal of $\mathcal{T}^*$, with terminals $t_1$, \ldots, $t_8$ indicated. Terminals $t_1$, $t_3$, $t_4$, $t_6$, $t_8$ have priority 2; the rest have priority 1.\vspace{-.4cm}}
\label{fig:dfs-traversal-example}
\end{figure}

As each edge in $\mathcal{T}^*$ is included twice in the traversal, the total weight of the edges visited in the traversal is $2\OPT$. Consider pairs of consecutive terminals visited for the first time along the traversal, including the pair consisting of the last and the first terminals visited. There are $|T|$ such terminal pairs; suppose these pairs are $(t_1,t_2)$, $(t_2,t_3)$, \ldots, $(t_{|T| - 1}, t_{|T|})$, $(t_{|T|}, t_1)$.

For $i \in [|T|]$, let $p_i^*$ denote the $t_i$--$t_{i+1}$ path in $\mathcal{T}^*$ using edges at their respective rates in $\mathcal{T}^*$, where $t_{|T|+1} := t_1$. Note that every edge in $p_i^*$ necessarily has rate at least $\min\{P(t_i), P(t_{i+1})\}$. Hence, if $c_i := \sum_{e \in p_i^*} w(e, R^*(e))$ denotes the weight of the edges along path $p_i$ in $\mathcal{T}^*$, we have $\sigma(t_i, t_{i+1}) \le c_i$. Further, the sum of the weights of these $|T|$ paths equals the weight of the edges in the traversal; that is, $\sum_{i=1}^{|T|} c_i = 2\OPT$. These observations imply $\sum_{i=1}^{|T|} \sigma(t_i, t_{i+1}) \le 2\OPT$.

Partition the $|T|$ terminal pairs into two disjoint sets $S_1$, $S_2$ of size $\frac{|T|}{2}$ as follows:
\begin{align*}
S_1 &= \{(t_1, t_2), (t_3, t_4), \ldots, (t_{|T|-1}, t_{|T|})\}\\
S_2 &= \{(t_2, t_3), (t_4, t_5), \ldots, (t_{|T|}, t_1)\}.
\end{align*}
Consider a pair $(t_i, t_{i+1}) \in S_1$. As Algorithm~\ref{alg:qosmt} connects terminals in decreasing order of priority, a candidate choice is to connect the lower-priority terminal to the higher-priority terminal, which implies that the connection cost of the lower-priority terminal is at most $\sigma(t_i, t_{i+1})$. Let $C_1 := \sigma(t_1, t_2) + \sigma(t_3, t_4) + \ldots + \sigma(t_{|T|-1}, t_{|T|})$; define $C_2$ similarly with respect to $S_2$. By considering all pairs in $S_1$, there necessarily exist $\frac{|T|}{2}$ terminals whose sum of connection costs is at most $C_1$. Similarly, there exist $\frac{|T|}{2}$ terminals whose sum of connection costs is at most $C_2$. As $C_1 + C_2 = \sum_{i=1}^{|T|} \sigma(t_i, t_{i+1}) \le 2\OPT$, either $C_1 \le \OPT$ or $C_2 \le \OPT$, so there exist $\frac{|T|}{2}$ terminals whose sum of connection costs is at most $\OPT$. $\qed$
\end{proof}

\begin{theorem} \label{thm:log2}
Algorithm~\ref{alg:qosmt} is a $(\lceil \log_2 |T|\rceil + 1)$-approximation for \textsc{Priority Steiner tree}.
\end{theorem}
\begin{proof}
By Lemma~\ref{lemma:pst-cc}, the sum of the $\frac{|T|}{2}$ cheapest connection costs is at most $\OPT$. Consider the remaining $\frac{|T|}{2}$ terminals with the most expensive connection costs, as well as the minimum subtree of $\mathcal{T}^*$ spanning these terminals. Applying Lemma~\ref{lemma:pst-cc} again, the sum of the next $\frac{|T|}{4}$ cheapest connection costs (out of these remaining terminals) is at most $\OPT$. We can apply Lemma~\ref{lemma:pst-cc} $\lceil\log_2 |T|\rceil + 1$ times to obtain the result. $\qed$
\end{proof}
In the following, we give a simpler, parallelizable algorithm for \textsc{Priority Steiner tree} which achieves the same approximation ratio of $\lceil\log_2 |T|\rceil+1$. For simplicity we assume $P(s)=\infty$ and every (non-source) terminal has a different priority; ties between terminals of the same priority can be broken arbitrarily. The idea is to connect each terminal $v$ to the ``closest'' terminal or source with a greater priority than $v$. Specifically, for $v \in T$, find a vertex $u \in T \cup \{s\}$ with $P(u) > P(v)$ which minimizes $\sigma(u,v)$, and connect $v$ to $u$ with edges of rate $P(v)$. This can be done by executing Dijkstra's algorithm from $v$ using edge weights $w(\cdot, P(v))$ and stopping once we find a vertex with a greater priority than $v$. Moreover, this algorithm is parallelizable as the corresponding path for each terminal can be found in parallel. The weight of connecting $v$ to its parent $u$ is the \emph{connection cost} of $v$. As before, cycles can be removed in the end by removing an edge from each cycle with the lowest rate.

\begin{algorithm}[h!]
\caption{$R(\cdot) = \text{PST}(\text{graph }G, \text{ priorities }P, \text{ edge weights }w, \text{ source } s)$}\label{alg:pst}
\begin{algorithmic}[1]
\State Initialize $R(e) = 0$ for $e \in E$
\For{$v \in T$}
\State Find $u \in T \cup \{s\}$ with $P(u) > P(v)$ such that $\sigma(u,v)$ is minimized
\State $R(e) = \max\{R(e), P(v)\}$ for $e \in p_{vu}$
\EndFor
\State Remove lowest-rate edge from each cycle
\State \Return{edge rates $R(\cdot)$}
\end{algorithmic}
\end{algorithm}
Algorithm~\ref{alg:pst} produces a valid PST which spans all terminals, since there is a path from each terminal $v$ to the source using edges of rate $P(v)$ or higher. Moreover, Lemma~\ref{lemma:qosmt-cc} and Theorem~\ref{thm:log2} extend easily:

\begin{lemma}
The sum of the $\frac{|T|}{2}$ least expensive connection costs incurred by Algorithm~\ref{alg:pst} is at most $\OPT$.
\end{lemma}
This is proved in the same way as Lemma~\ref{lemma:qosmt-cc}.
\begin{theorem}\label{thm:log2-pst}
Algorithm~\ref{alg:pst} is a $(\lceil\log_2 |T|\rceil + 1)$-approximation for \textsc{Priority Steiner tree}.
\end{theorem}
One main difference compared to Algorithm~\ref{alg:qosmt}~\cite{Charikar2004} is that Algorithm~\ref{alg:pst} is not required to connect the terminals sequentially, or even by order of priority. Further, unlike Algorithm~\ref{alg:qosmt}, Algorithm~\ref{alg:pst} is not dependent on the solution computed at the previous iteration. If $k \ll |T|$, a simple $k\rho$-approximation given by Charikar et al.~\cite{Charikar2004} is to compute a $\rho$-approximate Steiner tree over the terminals of each priority separately, taking $O(k \cdot \text{STEINER}(n))$ time. Executing both approximations and taking the better of the two solutions yields a $\min\{\lceil \log_2 |T|\rceil + 1, k\rho\}$-approximation as desired. The approximation ratios given in Theorem~\ref{thm:log2}--\ref{thm:log2-pst} are tight up to a factor of 2, even if $k=1$. Imase and Waxman~\cite{imase91dst} provide a sequence $(G_i)$ of graphs for which Algorithms~\ref{alg:qosmt}-\ref{alg:pst} may return a PST with weight $\frac{1}{2} \log_2 |T| + 1$ times the optimum; this was also analyzed recently in~\cite{ahmed2020kruskalbased}.

\section{An $O(\log |T|)$-approximation for Priority NWST} \label{section:greedyPNWST}
We remark that the analysis of Algorithms~\ref{alg:qosmt}-\ref{alg:pst}  does not extend to \textsc{Priority NWST}; one can construct an example input graph in which Algorithm~\ref{alg:qosmt} or~\ref{alg:pst} (considering minimum weight node-weighted paths) returns a poor NWST with weight $\Omega(|T|)\OPT$. In this section, we extend the $(2\ln |T|)$-approximation by Klein and Ravi~\cite{KLEIN1995104} which maintains a collection of trees, and greedily merges a subset of these trees at each iteration to minimize a cost-to-connectivity ratio (Algorithm~\ref{alg:greedyPNWST}). For \textsc{Priority NWST}, we need to ensure that the priority constraint is always maintained throughout the construction process. To this end, we first define a {\em rate tree}:

\begin{definition}[Rate tree]\label{def:rate-tree}
Let $G=(V,E)$, and let $\mathcal{T}_r$ be a subtree of $G$ (not necessarily a Steiner or spanning tree of $G$) which includes vertex $r$. Let $R:V \to \mathbb{R}_{\ge 0}$ be a function which assigns rates to the vertices in $G$. We say that $\mathcal{T}_r$ is a \emph{rate tree} rooted at $r$ if, for all $v \in V(\mathcal{T}_r) \setminus \{r\}$, the path from $r$ to $v$ in $\mathcal{T}_r$ consists of vertices of non-increasing rate.
\end{definition}



The main idea of Algorithm~\ref{alg:greedyPNWST} is to maintain a set (not necessarily a forest) of rate trees. By simply connecting the {\em roots} of the rate trees with paths of appropriate vertex rates, we can satisfy the priority constraints.


Another challenge to tackle involves properly devising a definition of weight when greedily merging rate trees at each iteration. The greedy NWST algorithm by Klein and Ravi~\cite{KLEIN1995104} simply sums the weights from a root vertex to each terminal. In our algorithm, we cannot simply connect the root of a rate tree to other roots of other rate trees of lower or equal priority and compute the weight similarly. This is due to a technical challenge needed for the analysis of the algorithm (see Section \ref{sec:analysisAlg2}) that it is not possible, in general, to perform a spider decomposition (similar to \cite{KLEIN1995104}) on a rate tree such that paths from the center to leaves have non-increasing rates. To overcome this challenge, we introduce the notion of {\em rate spiders} and prove the existence of a {\em rate spider decomposition}, which further guides us to properly define weight computations at each iterative step.

\subsection{Algorithm description}
In the following, let $p_1 < p_2 < \ldots < p_k$ denote the $k$ vertex priorities. Initialize a set $\mathcal{F}$ (not necessarily a forest) of $|T|+1$ rate trees so that each terminal $v \in T$, including the source $s$, is a singleton rate tree whose root is itself. Initialize vertex rates $R(v) = P(v)$ for $v \in T$, $R(s) = p_k$, and $R(v) = 0$ for $v \not\in T \cup \{s\}$. While $|\mathcal{F}| > 1$, the construction proceeds iteratively as follows. Each iteration consists of greedily selecting the following:
\begin{itemize}
    \item a rate tree $\mathcal{T}_r \in \mathcal{F}$ rooted at $r$, called the \emph{root tree}
    \item a special vertex $v \in V$ called the \emph{center} (note $v$ could equal $r$)
    \item a real number $b \le P(r)$ representing the rate which $v$ is  ``upgraded'' to
    \item a nonempty subset $\mathcal{S} = \{\mathcal{T}_{r_1}, \ldots, \mathcal{T}_{r_{|\mathcal{S}|}}\} \subset \mathcal{F}$ of rate trees where $\mathcal{T}_r \not\in \mathcal{S}$, and $P(r_j) \le b$ for all roots $r_j$ associated with the rate trees in $\mathcal{S}$
\end{itemize}
By connecting $r$ to the center $v$ using vertices of rate $b$, upgrading $R(v)$ to $b$, then connecting $v$ to the root of each rate tree $\mathcal{T}_{r_j} \in \mathcal{S}$ using vertices of rate $P(r_j)$, we can replace the $|\mathcal{S}|+1$ rate trees in $\mathcal{F}$ with a new rate tree $\mathcal{T}^{\text{new}}_r$ rooted at $r$ (see Figure~\ref{fig:pnwst-iteration-example}).
\begin{figure}[h]
\centering
\tikzset{every picture/.style={line width=0.75pt}} 

\begin{tikzpicture}[x=0.7pt,y=0.7pt,yscale=-.7,xscale=.7]

\draw  [color={rgb, 255:red, 0; green, 0; blue, 0 }  ,draw opacity=0.4 ] (167,63) .. controls (167,51.95) and (193.01,43) .. (225.1,43) .. controls (257.19,43) and (283.2,51.95) .. (283.2,63) .. controls (283.2,74.05) and (257.19,83) .. (225.1,83) .. controls (193.01,83) and (167,74.05) .. (167,63) -- cycle ;
\draw    (167,49) ;
\draw  [fill={rgb, 255:red, 0; green, 0; blue, 0 }  ,fill opacity=1 ] (220,63) .. controls (220,60.18) and (222.28,57.9) .. (225.1,57.9) .. controls (227.92,57.9) and (230.2,60.18) .. (230.2,63) .. controls (230.2,65.82) and (227.92,68.1) .. (225.1,68.1) .. controls (222.28,68.1) and (220,65.82) .. (220,63) -- cycle ;
\draw  [fill={rgb, 255:red, 0; green, 0; blue, 0 }  ,fill opacity=1 ] (414,33) .. controls (414,30.18) and (416.28,27.9) .. (419.1,27.9) .. controls (421.92,27.9) and (424.2,30.18) .. (424.2,33) .. controls (424.2,35.82) and (421.92,38.1) .. (419.1,38.1) .. controls (416.28,38.1) and (414,35.82) .. (414,33) -- cycle ;
\draw  [color={rgb, 255:red, 0; green, 0; blue, 0 }  ,draw opacity=0.4 ] (374.75,33) .. controls (374.75,20.63) and (394.61,10.6) .. (419.1,10.6) .. controls (443.59,10.6) and (463.45,20.63) .. (463.45,33) .. controls (463.45,45.37) and (443.59,55.4) .. (419.1,55.4) .. controls (394.61,55.4) and (374.75,45.37) .. (374.75,33) -- cycle ;
\draw  [fill={rgb, 255:red, 0; green, 0; blue, 0 }  ,fill opacity=1 ] (433,83) .. controls (433,80.18) and (435.28,77.9) .. (438.1,77.9) .. controls (440.92,77.9) and (443.2,80.18) .. (443.2,83) .. controls (443.2,85.82) and (440.92,88.1) .. (438.1,88.1) .. controls (435.28,88.1) and (433,85.82) .. (433,83) -- cycle ;
\draw  [color={rgb, 255:red, 0; green, 0; blue, 0 }  ,draw opacity=0.4 ] (397.5,83) .. controls (397.5,71.95) and (415.68,63) .. (438.1,63) .. controls (460.52,63) and (478.7,71.95) .. (478.7,83) .. controls (478.7,94.05) and (460.52,103) .. (438.1,103) .. controls (415.68,103) and (397.5,94.05) .. (397.5,83) -- cycle ;
\draw  [color={rgb, 255:red, 0; green, 0; blue, 0 }  ,draw opacity=0.4 ] (380.5,125) .. controls (380.5,113.95) and (398.68,105) .. (421.1,105) .. controls (443.52,105) and (461.7,113.95) .. (461.7,125) .. controls (461.7,136.05) and (443.52,145) .. (421.1,145) .. controls (398.68,145) and (380.5,136.05) .. (380.5,125) -- cycle ;
\draw  [fill={rgb, 255:red, 0; green, 0; blue, 0 }  ,fill opacity=1 ] (408.9,127) .. controls (408.9,125.56) and (410.06,124.4) .. (411.5,124.4) .. controls (412.94,124.4) and (414.1,125.56) .. (414.1,127) .. controls (414.1,128.44) and (412.94,129.6) .. (411.5,129.6) .. controls (410.06,129.6) and (408.9,128.44) .. (408.9,127) -- cycle ;
\draw  [fill={rgb, 255:red, 0; green, 0; blue, 0 }  ,fill opacity=1 ] (256,58) .. controls (256,55.18) and (258.28,52.9) .. (261.1,52.9) .. controls (263.92,52.9) and (266.2,55.18) .. (266.2,58) .. controls (266.2,60.82) and (263.92,63.1) .. (261.1,63.1) .. controls (258.28,63.1) and (256,60.82) .. (256,58) -- cycle ;
\draw  [fill={rgb, 255:red, 0; green, 0; blue, 0 }  ,fill opacity=1 ] (284,75) .. controls (284,72.18) and (286.28,69.9) .. (289.1,69.9) .. controls (291.92,69.9) and (294.2,72.18) .. (294.2,75) .. controls (294.2,77.82) and (291.92,80.1) .. (289.1,80.1) .. controls (286.28,80.1) and (284,77.82) .. (284,75) -- cycle ;
\draw    (225.1,63) -- (261.1,58) ;
\draw    (261.1,58) -- (289.1,75) ;
\draw    (289.1,75) -- (316.1,79) ;
\draw  [fill={rgb, 255:red, 0; green, 0; blue, 0 }  ,fill opacity=1 ] (311,79) .. controls (311,76.18) and (313.28,73.9) .. (316.1,73.9) .. controls (318.92,73.9) and (321.2,76.18) .. (321.2,79) .. controls (321.2,81.82) and (318.92,84.1) .. (316.1,84.1) .. controls (313.28,84.1) and (311,81.82) .. (311,79) -- cycle ;
\draw  [fill={rgb, 255:red, 0; green, 0; blue, 0 }  ,fill opacity=1 ] (344,85) .. controls (344,82.18) and (346.28,79.9) .. (349.1,79.9) .. controls (351.92,79.9) and (354.2,82.18) .. (354.2,85) .. controls (354.2,87.82) and (351.92,90.1) .. (349.1,90.1) .. controls (346.28,90.1) and (344,87.82) .. (344,85) -- cycle ;
\draw  [fill={rgb, 255:red, 0; green, 0; blue, 0 }  ,fill opacity=1 ] (359.6,114.8) .. controls (359.6,113.36) and (360.76,112.2) .. (362.2,112.2) .. controls (363.64,112.2) and (364.8,113.36) .. (364.8,114.8) .. controls (364.8,116.24) and (363.64,117.4) .. (362.2,117.4) .. controls (360.76,117.4) and (359.6,116.24) .. (359.6,114.8) -- cycle ;
\draw    (316.1,79) -- (349.1,85) ;
\draw    (349.1,85) -- (362.2,114.8) ;
\draw    (362.2,114.8) -- (411.5,127) ;
\draw  [fill={rgb, 255:red, 0; green, 0; blue, 0 }  ,fill opacity=1 ] (342.1,50.8) .. controls (342.1,47.98) and (344.38,45.7) .. (347.2,45.7) .. controls (350.02,45.7) and (352.3,47.98) .. (352.3,50.8) .. controls (352.3,53.62) and (350.02,55.9) .. (347.2,55.9) .. controls (344.38,55.9) and (342.1,53.62) .. (342.1,50.8) -- cycle ;
\draw    (316.1,79) -- (347.2,50.8) ;
\draw    (347.2,50.8) -- (419.1,33) ;
\draw    (349.1,85) -- (443.2,83) ;
\draw  [color={rgb, 255:red, 0; green, 0; blue, 0 }  ,draw opacity=0.4 ][dash pattern={on 0.84pt off 2.51pt}] (367,34.16) .. controls (367,18.06) and (380.06,5) .. (396.16,5) -- (485.04,5) .. controls (501.14,5) and (514.2,18.06) .. (514.2,34.16) -- (514.2,121.64) .. controls (514.2,137.74) and (501.14,150.8) .. (485.04,150.8) -- (396.16,150.8) .. controls (380.06,150.8) and (367,137.74) .. (367,121.64) -- cycle ;

\draw (150,56.4) node [anchor=north west][inner sep=0.75pt]    {$\mathcal{T}_{r}$};
\draw (219,70.4) node [anchor=north west][inner sep=0.75pt]    {$r$};
\draw (275.1,82.5) node [anchor=north west][inner sep=0.75pt]    {$ \begin{array}{l}
\ \ \ \ \ v\\
( b=2)
\end{array}$};
\draw (426,23.4) node [anchor=north west][inner sep=0.75pt]    {$r_{1}$};
\draw (445,72.4) node [anchor=north west][inner sep=0.75pt]    {$r_{2}$};
\draw (418,117.4) node [anchor=north west][inner sep=0.75pt]    {$r_{3}$};
\draw (419,154.4) node [anchor=north west][inner sep=0.75pt]    {$\mathcal{S}$};
\draw (470,22.4) node [anchor=north west][inner sep=0.75pt]    {$\mathcal{T}_{r_1}$};
\draw (484,74.4) node [anchor=north west][inner sep=0.75pt]    {$\mathcal{T}_{r_2}$};
\draw (465,117.4) node [anchor=north west][inner sep=0.75pt]    {$\mathcal{T}_{r_3}$};

\end{tikzpicture}
\caption{Illustration of an iteration step in Algorithm~\ref{alg:greedyPNWST} with $P(r)=2$, $b=2$, $P(r_1) = P(r_2) = 2$, and $P(r_3) = 1$. Vertices with larger circles (not necessarily terminals) have rate 2; vertices with smaller circles have rate 1.}
\label{fig:pnwst-iteration-example}
\end{figure}
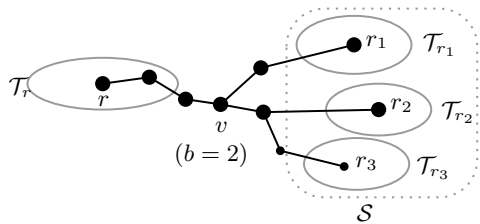

The root tree, center, $b$, and $\mathcal{S}$ are greedily chosen to minimize a cost-to-connectivity ratio $\gamma$, defined as follows:
\begin{equation} \label{eq:gamma}
    \gamma := \frac{1}{|\mathcal{S}|+1}\left(\sigma_b(r,v) + w(v,b) + \sum_{j=1}^{|\mathcal{S}|}\sigma_{P(r_j)}(v, r_j)\right)
\end{equation}
where $r_j$ denotes the root of the $j^{\text{th}}$ rate tree $\mathcal{T}_{r_j}$ in $\mathcal{S}$. 
The second expression $\sigma_b(r,v) + w(v,b) + \sum_{j=1}^{|\mathcal{S}|} \sigma_{P(r_j)}(v,r_j)$ gives an upper bound on the weight of connecting $r$ to $v$, upgrading $R(v)$ to $b$, then connecting $v$ to $|\mathcal{S}|$ roots, and the denominator $|\mathcal{S}|+1$ represents the ``connectivity'', or the number of connected rate trees. Lemma~\ref{lemma:PNWST-polytime} shows how to execute this iteration step in polynomial time.

Once $\mathcal{T}_r$, $v$, $b$, and $\mathcal{S}$ are chosen, we ``upgrade'' the vertex rates $R(\cdot)$ along a shortest $r$--$v$ path to $b$, then upgrade the vertex rates along each shortest $v$--$r_j$ path to $P(r_j)$. In the case that some vertex $u$ is on multiple $v$--$r_j$ paths, then $R(u)$ is upgraded to the maximum over all root priorities $P(r_j)$ for which $u$ appears on the corresponding path. Pseudocode is shown in Algorithm~\ref{alg:greedyPNWST}.

\begin{algorithm}[ht]
\caption{$R(\cdot) = \text{PNWST}(G, \text{ terminals } T, \text{ priorities } P, \text{ vertex weights } w)$}
\begin{algorithmic}[1]
\State Initialize $\mathcal{F}$, $R(v) = P(v)$ if $v \in T \cup \{s\}$ and $R(v) = 0$ if $v \not\in T \cup \{s\}$
\While{$|\mathcal{F}|>1$}
\State Find $\mathcal{T}_r$, $v$, $b$, $\mathcal{S}$ which minimize $\gamma$ (Lemma~\ref{lemma:PNWST-polytime})
\State $R(u) = \max\{R(u), b\}$ for $u$ on $r$--$v$ path \label{line:pnwst-shortest}
\State $R(v) = \max\{R(v), b\}$
\For {$j = 1, \ldots, |\mathcal{S}|$}
\State $R(u) = \max\{R(u), P(r_j)\}$ for $u$ on $v$--$r_j$ path \label{line:pnwst-upgrade}
\EndFor
\State $\mathcal{F} = \mathcal{F} \setminus (\{\mathcal{T}_r\} \cup \mathcal{S})$
\State $\mathcal{F} = \mathcal{F} \cup \{\mathcal{T}^{\text{new}}_r$\}
\EndWhile
\State \Return{vertex rates $R(\cdot)$}
\end{algorithmic}\label{alg:greedyPNWST}
\end{algorithm}

\subsection{Analysis of Algorithm~\ref{alg:greedyPNWST}}\label{sec:analysisAlg2}
We show Theorem~\ref{thm:greedyPNWST} by asserting that Algorithm~\ref{alg:greedyPNWST} is a $2 \ln (|T|+1)$--approximation for \textsc{Priority NWST}. Proofs omitted due to space are in the arXiv version{\color{red}cite}

We extend the spider decomposition given by Klein and Ravi~\cite{KLEIN1995104} to account for the priority constraints in the \textsc{Priority NWST} problem.

\begin{definition}[Spider]
A \emph{spider} is a tree where at most one vertex has degree greater than 2. A \emph{nontrivial spider} is a spider with at least 2 leaves.
\end{definition}
A spider is identified by its \emph{center}, a vertex from which all paths from the center to the leaves of the spider are vertex-disjoint. A foot of a spider is a leaf; if the spider has at least three leaves, then its center is unique and is also a foot. Klein and Ravi~\cite{KLEIN1995104} show that given a graph $G$ and subset $M \subseteq V$ of vertices, $G$ can be decomposed into vertex-disjoint nontrivial spiders such that the union of the feet of the nontrivial spiders contains $M$. We extend the notions of spider and spider decomposition to the \textsc{Priority NWST} problem.

\begin{definition}[Rate spider]
A \emph{rate spider} is a rate tree $\mathcal{X}$ which is also a nontrivial spider. It is identified by a \emph{root} $r$ as well as a \emph{center} $v$ such that:
\begin{itemize}
    \item The root $r$ is either the center or a leaf of $\mathcal{X}$, and the path from $r$ to every vertex in $\mathcal{X}$ uses vertices of non-increasing rate $R(\cdot)$
    \item The paths from the center $v$ to each \emph{non-root} leaf of $\mathcal{X}$ are vertex-disjoint and use vertices of non-increasing rate $R(\cdot)$.
\end{itemize}
\end{definition}
In Figure~\ref{fig:spider-decomposition-1}, right, rate spiders $\mathcal{X}_2$ and $\mathcal{X}_3$ have centers distinct from their roots $r_2$, $r_3$ while $\mathcal{X}_1$ has center $v = r_1$. In Definition~\ref{def:M-optimized}, we supply a notion of a ``minimal'' weight tree with respect to a subset $M$ of vertices.

\begin{definition}[$M$--optimized rate tree] \label{def:M-optimized}
Let $\mathcal{T}_r$ be a rate tree rooted at $r$ with vertex rates $R$. Let $M \subseteq V(\mathcal{T}_r)$ with $r \in M$. Then $\mathcal{T}_r$ is \emph{$M$--optimized} if every leaf of $\mathcal{T}_r$ is in $M$, and if for every vertex $v \in V(\mathcal{T}_r) \setminus M$, we have $R(v) = \max R(w)$ over all vertices $w \in M$ in the subtree of $\mathcal{T}_r$ rooted at $v$.
\end{definition}

We show any $M$--optimized rate tree has a rate spider decomposition.

\begin{lemma}[Rate spider decomposition] \label{lemma:decomposition}
Let $M \subseteq V(\mathcal{T}_r)$ with $|M| \ge 2$, and let $\mathcal{T}_r$ be an $M$--optimized rate tree where $r \in M$. Then $\mathcal{T}_r$ can be decomposed into vertex-disjoint rate spiders $\mathcal{X}_1$, \ldots, $\mathcal{X}_d$ rooted at $r_1$, \ldots, $r_d$ such that:
\begin{itemize}
    \item the leaves and roots of the rate spiders are contained in $M$
    \item every vertex in $M$ is a either a leaf, root, or center of some rate spider
\end{itemize}
\end{lemma}
Figure~\ref{fig:spider-decomposition-1}, right, shows an example of an $M$--optimized rate tree $\mathcal{T}_r$ for $|M|=10$ and a rate spider decomposition $\mathcal{X}_1$, $\mathcal{X}_2$, $\mathcal{X}_3$ over $M$.

\begin{figure}[ht]
\centering
\begin{subfigure}[b]{0.49\textwidth}
\begin{tikzpicture}[scale=0.5, yscale=-1,
every node/.style={circle,fill=black,
draw=black, inner sep=0pt, minimum size=4pt}]
\node[label={{right,above=4pt}:{\footnotesize 3}}, label={{below}:{\footnotesize $r$}}] (r) at (0,0) {};
\node[label={left,above=4pt}:{\footnotesize 2}] (1a) at (-2,1) {};
\node[label={right,above=4pt}:{\footnotesize 3},fill=white] (1b) at (2,1) {};
\node[label={left,above=4pt}:{\footnotesize 2},fill=white] (2a) at (-3,2) {};
\node[label={right,above=4pt}:{\footnotesize 2},fill=white] (2b) at (-2,2) {};
\node[label={right,above=4pt}:{\footnotesize 2}] (2c) at (-1,2) {};
\node[label={left,above=4pt}:{\footnotesize 3},fill=white] (2d) at (1,2) {};
\node[label={right,above=4pt}:{\footnotesize 3},fill=white] (2e) at (2.5,2) {};
\node[label={left,above=4pt}:{\footnotesize 1}] (3a) at (-3.5,3) {};
\node[label={left,above=4pt}:{\footnotesize 2},fill=white] (3b) at (-2.5,3) {};
\node[label={right,above=4pt}:{\footnotesize 2}] (3c) at (-1.5,3) {};
\node[label={left,above=4pt}:{\footnotesize 2}] (3d) at (0.5,3) {};
\node[label={right,above=4pt}:{\footnotesize 3},fill=white] (3e) at (1.5,3) {};
\node[label={right,above=4pt}:{\footnotesize 1},fill=white] (3f) at (3,3) {};
\node[label={left,above=4pt}:{\footnotesize 1}] (4a) at (-2.5,4) {};
\node[label={right,above=4pt}:{\footnotesize 2}] (4b) at (-0.5,4) {};
\node[label={left,above=4pt}:{\footnotesize 1},fill=white] (4c) at (1,4) {};
\node[label={right,above=4pt}:{\footnotesize 1}] (4d) at (2,4) {};
\node[label={right,above=4pt}:{\footnotesize 1}] (4e) at (3.5,4) {};

\draw (r)--(1a); \draw (r)--(1b);
\draw (1a)--(2a); \draw (1a)--(2b); \draw (1a)--(2c);
\draw (1b)--(2d); \draw (1b)--(2e);
\draw (2a)--(3a);
\draw (2b)--(3b);\draw (2b)--(3c);
\draw (2d)--(3d);\draw (2d)--(3e);
\draw (2d)--(3f);
\draw (3c)--(4a);\draw (3c)--(4b);
\draw (3e)--(4c);\draw (3e)--(4d);\draw (3f)--(4e);
\end{tikzpicture}
\end{subfigure}
\begin{subfigure}[b]{0.49\textwidth}
\begin{tikzpicture}[scale=0.5, yscale=-1,
every node/.style={circle,fill=black,
draw=black, inner sep=0pt, minimum size=4pt}]
\node[label={right,above=4pt}:{\footnotesize 3},label={{below}:{\footnotesize $r_3$}}] (r) at (0,0) {};
\node[label={left,above=4pt}:{\footnotesize 2}] (1a) at (-2,1) {};
\node[label={right,above=4pt}:{\footnotesize 2},fill=white] (1b) at (2,1) {};
\node[label={left,above=4pt}:{\footnotesize 1},fill=white] (2a) at (-3,2) {};
\node[label={right,above=4pt}:{\footnotesize 2},fill=white] (2b) at (-2,2) {};
\node[label={right,above=4pt}:{\footnotesize 2}] (2c) at (-1,2) {};
\node[label={left,above=4pt}:{\footnotesize 2},fill=white] (2d) at (1,2) {};

\node[label={left,above=4pt}:{\footnotesize 1}] (3a) at (-3.5,3) {};

\node[label={right,above=4pt}:{\footnotesize 2},label={{below}:{\footnotesize $r_1$}}] (3c) at (-1.5,3) {};
\node[label={left,above=4pt}:{\footnotesize 2},label={{below}:{\footnotesize $r_2$}}] (3d) at (0.5,3) {};
\node[label={right,above=4pt}:{\footnotesize 1},fill=white] (3e) at (1.5,3) {};
\node[label={right,above=4pt}:{\footnotesize 1},fill=white] (3f) at (3,3) {};
\node[label={left,above=4pt}:{\footnotesize 1}] (4a) at (-2.5,4) {};
\node[label={right,above=4pt}:{\footnotesize 2}] (4b) at (-0.5,4) {};

\node[label={right,above=4pt}:{\footnotesize 1}] (4d) at (2,4) {};
\node[label={right,above=4pt}:{\footnotesize 1}] (4e) at (3.5,4) {};

\draw[ultra thick] (r)--(1a); \draw (r)--(1b);
\draw[ultra thick] (1a)--(2a); \draw (1a)--(2b); \draw[ultra thick] (1a)--(2c);
\draw (1b)--(2d);
\draw[ultra thick] (2a)--(3a);
\draw (2b)--(3c);
\draw[ultra thick] (2d)--(3d);\draw[ultra thick] (2d)--(3e);
\draw[ultra thick] (2d)--(3f);

\draw[ultra thick] (3c)--(4a);\draw[ultra thick] (3c)--(4b);
\draw[ultra thick] (3e)--(4d);
\draw[ultra thick] (3f)--(4e);
\node[fill=white,draw=white,minimum size=0pt] at (-2.2,3) {$\mathcal{X}_1$};
\node[fill=white,draw=white,minimum size=0pt] at (2.5,2) {$\mathcal{X}_2$};
\node[fill=white,draw=white,inner sep=0pt,minimum size=0pt] at (-1.5,0.2) {$\mathcal{X}_3$};
\end{tikzpicture}
\end{subfigure}
\caption{\emph{Left:} A rate tree rooted at $r$ with rates $R(\cdot)$ indicated and vertices in $M$ shown in black. \emph{Right:} An $M$--optimized rate tree $\mathcal{T}_r$ and a rate spider decomposition $\mathcal{X}_1$, $\mathcal{X}_2$, $\mathcal{X}_3$ with roots $r_1$, $r_2$, $r_3$.}
\label{fig:spider-decomposition-1}
\end{figure}
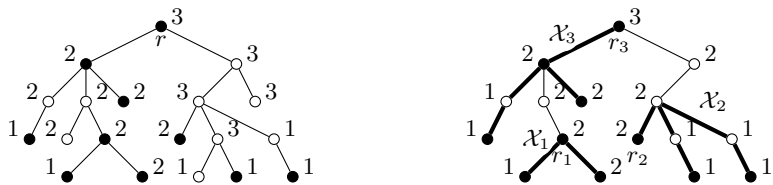

\begin{proof}
We use induction on $|M|$. For base case $|M|=2$, the decomposition consists of a single rate spider $\mathcal{X}_1$ with root $r_1 = r$, namely the $r$--$u$ path in $\mathcal{T}_r$ where $u$ is the other vertex in $M$. Suppose that for some $q \in \{2,\ldots,|V(\mathcal{T}_r)|-1\}$, there exists a rate spider decomposition over any subset $M' \subset V(\mathcal{T}_r)$ with $2 \le |M'| \le q$, so that we wish to show that there exists a decomposition of $\mathcal{T}_r$ over any subset $M$ with $|M|=q+1$.

Given $\mathcal{T}_r$ rooted at $r$ and $M$, find a vertex $u \in V(\mathcal{T}_r)$ furthest from $r$ (by number of edges) with the property that the subtree $\mathcal{T}_u$ rooted at $u$ contains at least two vertices in $M$. If $u=r$, then we claim that $\mathcal{T}_r$ is already a rate spider with root and center $r$, in which a decomposition of $\mathcal{T}_r$ over $M$ is itself. To show this claim, we show that every $w \in V(\mathcal{T}_r)$ with $w \neq r$ has degree at most 2 in $\mathcal{T}_r$. Suppose otherwise there exists some $w \neq r$ with degree at least 3 in $\mathcal{T}_r$. Then the subtree $\mathcal{T}_w$ rooted at $w$ contains at least two leaves which are contained in $M$ (by Def.~\ref{def:M-optimized}), contradicting the choice of $u=r$ since $w$ is further from $r$.

If $u \neq r$, then the subtree $\mathcal{T}_u$ rooted at $u$ is a rate spider $\mathcal{X}$ with center $u$. If $u \in M$, set $u$ as the root of $\mathcal{X}$. If $u \not\in M$, find a vertex $u' \in V(\mathcal{T}_u)$ with rate $R(u)$; such a vertex $u'$ exists by Def.~\ref{def:M-optimized}. Set $u'$ to be the root of $\mathcal{X}$; remove $\mathcal{X}$ from $\mathcal{T}_r$ as well as the edge from $u$ to its parent to produce a smaller rate tree $\mathcal{T}'_r$.

Let $M' = M \cap V(\mathcal{T}'_r)$ be the set of vertices in $M$ which remain in $\mathcal{T}_r$ after removing $\mathcal{T}_u$. If $|M'| = 0$, then we have found a rate spider decomposition of $\mathcal{T}_r$ over $M$. If $|M'| = 1$, then connecting $r$ to $\mathcal{X}$ yields a single rate spider rooted at $r$ with center $u$, which also gives a rate spider decomposition. If $|M'| \ge 2$, then prune the $r$--$u$ path so that $\mathcal{T}'_r$ is $M'$--optimized. By induction hypothesis, $\mathcal{T}'_r$ contains a rate spider decomposition over $M'$. $\qed$
\end{proof}

\begin{corollary}\label{corollary:m}
Let $\mathcal{T}_r$ be an $M$--optimized rate tree. Consider a rate spider decomposition of $\mathcal{T}_r$ over $M$ containing $d$ rate spiders $\mathcal{X}_1$, \ldots, $\mathcal{X}_d$ generated using the method in the proof of Lemma~\ref{lemma:decomposition}. Let $S_j = (M \cap V(\mathcal{X}_j)) \setminus \{r_j\}$ denote the vertices in $M$ contained in $\mathcal{X}_j$, not including its root $r_j$. Then $\sum_{j=1}^d (1 + |S_j|) = |M|$.
\end{corollary}
\begin{proof}
This immediately follows as the $j^{\text{th}}$ rate spider  $\mathcal{X}_j$ contains $1+|S_j|$ vertices in $M$. Since the rate spiders are vertex-disjoint and every vertex in $M$ is contained in some rate spider, it follows that $\sum_{j=1}^d (1 + |S_j|) = |M|$. $\qed$
\end{proof}

For the following lemmas, we define the following notation. For $i \ge 1$, let $\mathcal{F}_i$ denote the set of rate trees at the beginning of iteration $i$ in Algorithm~\ref{alg:greedyPNWST}, and let $M_i$ denote the set of roots of the rate trees in $\mathcal{F}_i$. Hence $|\mathcal{F}_1| = |M_1| = |T|+1$. Let $\mathcal{T}^*$ be a minimum weight PNWST for the instance with weight $\OPT$ and vertex rates $R^*: V(\mathcal{T}^*) \to \{p_1,\ldots,p_k\}$ where $p_1$, \ldots, $p_k$ are the $k$ vertex priorities.

On iteration $i$, let $\mathcal{T}^*_i$ be an $M_i$--optimized rate tree by optimizing $\mathcal{T}^*$ with respect to $M_i$. By Lemma~\ref{lemma:decomposition}, $\mathcal{T}^*_i$ contains a rate spider decomposition over $M_i$ containing $d$ rate spiders $\mathcal{X}_1$, \ldots, $\mathcal{X}_d$. For $j \in [d]$, consider the $j^{\text{th}}$ rate spider $\mathcal{X}_j$ with root $r_j$, center $v_j$, and leaves $S_j \subset M_i$ not including $r_j$. Let $w(\mathcal{X}_j)$ denote the weight of the vertices in $\mathcal{X}_j$ within the optimum solution $\mathcal{T}^*$, given by $w(\mathcal{X}_j) = \sum_{v \in V(\mathcal{X}_j)} w(v, R^*(v))$. On iteration $i$, a candidate choice which Algorithm~\ref{alg:greedyPNWST} could select is according to the $j^{\text{th}}$ rate spider $\mathcal{X}_j$: specifically, select root tree $\mathcal{T}_{r_j}$ rooted at $r_j$, center $v = v_j$, $b = R^*(v_j)$, and $\mathcal{S} \subseteq \mathcal{F}_i$ the set of rate trees whose root is in $S_j$. Let $c_j$ be the weight that Algorithm~\ref{alg:greedyPNWST} computes for this candidate choice (i.e., the second expression in Eq.~\eqref{eq:gamma}). We observe that for all rate spiders $\mathcal{X}_j$ in a rate spider decomposition of $\mathcal{T}_i^*$, we have $c_j \le w(\mathcal{X}_j)$; this follows as the computed weight $c_j$ considers the minimum weight vertex-weighted paths between $r_j$ and $v_j$, as well as from $v_j$ to each leaf in $S_j$.

Let $h_i \ge 2$ denote the number of rate trees in $\mathcal{F}_i$ which are selected on iteration $i$ of Algorithm~\ref{alg:greedyPNWST} (i.e., $h_i = |\mathcal{S}|+1$). Let $\Delta C_i$ denote the actual weight incurred on iteration $i$ by upgrading vertex rates in line~\ref{line:pnwst-upgrade}. Let $\gamma_i$ denote the minimum cost-to-connectivity ratio (Eq.~\eqref{eq:gamma}) computed by Algorithm~\ref{alg:greedyPNWST} on iteration $i$.

\begin{lemma}\label{lemma:iteration-i}
For each iteration $i$ of Algorithm~\ref{alg:greedyPNWST}, we have $\dfrac{\Delta C_i}{h_i} \le \dfrac{\OPT}{|\mathcal{F}_i|}$.
\end{lemma}

\begin{proof}
Consider the $M_i$--optimized rate tree obtained from $\mathcal{T}^*$. By Lemma~\ref{lemma:decomposition}, there exists a rate spider decomposition over $M_i$ containing $d \ge 1$ rate spiders $\mathcal{X}_1$, \ldots, $\mathcal{X}_d$.

For each $j \in [d]$, as Algorithm~\ref{alg:greedyPNWST} seeks to minimize $\gamma$, we have by the above observation that $c_j \le w(\mathcal{X}_j)$:
\begin{equation} \label{eq:gamma-i}
\gamma_i \le \frac{c_j}{1+|S_j|} \le \frac{w(\mathcal{X}_j)}{1+|S_j|}.
\end{equation}
We note that $\frac{\Delta C_i}{h_i} \le \gamma_i$; this follows as the computed weight in Algorithm~\ref{alg:greedyPNWST} may overcount vertex weights appearing on multiple center-to-root paths. This observation, combined with inequality~\eqref{eq:gamma-i}, implies $\frac{\Delta C_i}{h_i} \le \frac{w(\mathcal{X}_j)}{1 + |S_j|}$ for all rate spiders $\mathcal{X}_j$ in the decomposition.

We use the simple algebraic fact that for non-negative numbers $a$, $x_1$, \ldots, $x_d$, $y_1$, \ldots, $y_d$ where the $y_j$'s are nonzero, if $a \le \frac{x_j}{y_j}$ for all $j \in [d]$, then $a \le (\sum_{j=1}^d x_j)/(\sum_{j=1}^d y_j)$. This fact is easily verified by writing $ay_j \le x_j$, then summing from $j=1$ to $j=d$. Applying this fact, we obtain
\begin{equation}
\frac{\Delta C_i}{h_i} \le \gamma_i \le \frac{\sum_{j=1}^d w(\mathcal{X}_j)}{\sum_{j=1}^d 1+|S_j|} \le \frac{\OPT}{|\mathcal{F}_i|}
\end{equation}
where the last inequality follows from the fact that $\sum_{j=1}^d w(\mathcal{X}_j) \le \OPT$ (as the vertices in a rate spider decomposition of $\mathcal{T}^*$ are a subset of $V(\mathcal{T}^*)$), as well as Corollary~\ref{corollary:m}. $\qed$
\end{proof}

Using Lemma~\ref{lemma:iteration-i}, we can prove Theorem~\ref{thm:greedyPNWST}, by asserting that Algorithm~\ref{alg:greedyPNWST} is a $2 \ln (|T|+1)$--approximation for \textsc{Priority NWST}. The remainder of the proof can be completed by following the analysis by Klein and Ravi~\cite{KLEIN1995104}. We use a simpler analysis to show a marginally weaker approximation ratio.

\begin{proof}[Theorem~\ref{thm:greedyPNWST}]
Lemma~\ref{lemma:iteration-i} can equivalently be written as $\Delta C_i \le \frac{h_i}{|\mathcal{F}_i|}\OPT$.
Recall that $h_i \ge 2$ and $|\mathcal{F}_i|$ denote the number of rate trees merged on iteration $i$ and the number of rate trees in $\mathcal{F}$ at the beginning of iteration $i$, respectively. Suppose Algorithm~\ref{alg:greedyPNWST} runs for $I$ iterations. Thus we have $|\mathcal{F}_1| = |T|+1$ and $|\mathcal{F}_I| = 1$, as well as the relation $|\mathcal{F}_{i+1}| = |\mathcal{F}_i| - (h_i - 1) = |\mathcal{F}_i| - h_i + 1$. 

We use the simple algebraic fact that for positive integers $x < y$, we have $\frac{x}{y} \le \frac{1}{y} + \frac{1}{y-1} + \ldots + \frac{1}{y-x+1} = H_y - H_{y-x}$, where $H_x = \frac{1}{1} + \frac{1}{2} + \ldots + \frac{1}{x}$ is the $x^{\text{th}}$ harmonic number with $H_0 := 0$. Applying this fact, we have
\[\frac{h_i}{|\mathcal{F}_i|} \le H_{|\mathcal{F}_i|} - H_{|\mathcal{F}_i|-h_i} = H_{|\mathcal{F}_i|} - H_{|\mathcal{F}_{i+1}| - 1}. \]

Multiplying the above inequality by $\OPT$ and applying Lemma~\ref{lemma:iteration-i} yields
\begin{equation}
\Delta C_i \le (H_{|\mathcal{F}_i|} - H_{|\mathcal{F}_{i+1}| - 1}) \OPT. \label{eq:delta-ci}
\end{equation}

Summing inequality~\eqref{eq:delta-ci} from $i=1$ to $i=I$, we obtain an upper bound on the PNWST weight $\sum_{i=1}^{I} \Delta C_i$ in terms of $\OPT$ which approximately telescopes to $2 \ln (|T|+1) \OPT$. Specifically, let $F = \{ |\mathcal{F}_2|, \ldots, |\mathcal{F}_I|\}$ be a subset of $[|T|]$ where $1 \in F$ and $|T|+1 \not\in F$. Then summing inequality~\eqref{eq:delta-ci} equivalently yields
\begin{align*}
w(\mathcal{T}) = \sum_{i=1}^I \Delta C_i &\le H_{|T|+1}\OPT + \sum_{x \in F} (H_x - H_{x-1})\OPT\\
&\le (H_{|T|+1} + H_{|T|})\OPT \\
&\le (2 \ln (|T|+1) + 2)\OPT
\end{align*}
completing the proof. As stated previously, following the same analysis as in~\cite{KLEIN1995104} proves a marginally better approximation ratio of $2 \ln (|T|+1)$. $\qed$
\end{proof}

It is worth noting that the extension of the $(2 \ln |T|)$-approximation by Klein and Ravi~\cite{KLEIN1995104} to the \textsc{Priority NWST} problem is not immediately obvious, as we must be careful when merging multiple rate trees while simultaneously satisfying the priority and rate requirements.

\begin{lemma}\label{lemma:PNWST-polytime}
On iteration $i$ of Algorithm~\ref{alg:greedyPNWST}, a choice of $\mathcal{T}_r$, $v$, $b$, and $\mathcal{S}$ which minimizes $\gamma$ can be found in $O(n^3 k \log n)$ time.
\end{lemma}
\begin{proof}
For a fixed root tree $\mathcal{T}_r$, center $v$, and $b \in \{p_1,p_2,\ldots,P(r)\}$, consider the set $\mathcal{F}'$ of all rate trees in $\mathcal{F} \setminus \{\mathcal{T}_r\}$ whose root has priority at most $b$. Sort $\mathcal{F}'$ in nondecreasing order by the weight of a minimum weight path from $v$ to its root, namely $\sigma_{P(r')}(v, r')$ for rate tree $\mathcal{T}_{r'}$, taking $O(n \log n)$ time. Then for fixed $\mathcal{T}_r$, center $v$, $b$, we can determine $\mathcal{S}$ which minimizes $\gamma$ by only considering $\mathcal{S}$ to be the first 1, 2, 3, \ldots, $|\mathcal{F}'|$ trees in the sorted list.

There are $|\mathcal{F}| = O(n)$ possible root trees, $n$ possible centers, and $O(k)$ possible choices for $b$. Using the above analysis, we can determine $\mathcal{S}$ which minimizes $\gamma$ over all possible choices in $O(n^3 k \log n)$ time. $\qed$
\end{proof}

Algorithm~\ref{alg:greedyPNWST} runs for $I \le |T|$ iterations, as the size of $|\mathcal{F}|$ decreases by at least 1 at each iteration. By Lemma~\ref{lemma:PNWST-polytime}, the running time of Algorithm~\ref{alg:greedyPNWST} is $O(n^4 k \log n)$. The approximation ratio for Algorithm~\ref{alg:greedyPNWST} is tight as is the case for the Ravi-Klein algorithm~\cite{KLEIN1995104}; see Figure~\ref{fig:pnwst-tightness}.

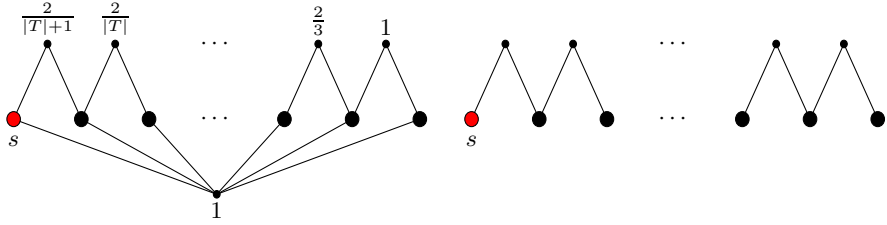
\begin{figure}[h]
\centering
\begin{subfigure}[b]{0.49\textwidth}
\centering
	\begin{tikzpicture}[xscale=0.9]
	\foreach \x in {-1.5, -0.5, 2.5, 3.5}{
	    \draw [fill=black] (\x, 1) circle [radius=0.05];
	    \draw (\x, 1) -- (\x-0.5, 0);
	    \draw (\x, 1) -- (\x+0.5, 0);
	}
	
	\foreach \x in {-2,-1,0,2,3,4}{
	    \draw (1, -1) -- (\x, 0);
		\ifthenelse{\x=-2}
		{\draw [fill=red] (\x, 0) circle [radius=0.1];}
		{\draw [fill=black] (\x, 0) circle [radius=0.1];}
		
	}

	\draw [fill=black] (1, -1) circle [radius=0.05];
	\node[below] at (1, -1) {\small 1};
	\node[above] at (-1.5, 1) {\small $\frac{2}{|T|+1}$};
	\node[above] at (-0.5, 1) {\small $\frac{2}{|T|}$};
	\node[above] at (2.5, 1) {\small $\frac{2}{3}$};
	\node[above] at (3.5, 1) {\small 1};
	\node at (1,0) {$\cdots$};
	\node at (1,1) {$\cdots$};
	\node[below] at (-2,-0.1) {$s$};
	\end{tikzpicture}
\end{subfigure}
\begin{subfigure}[b]{0.49\textwidth}
\centering
	\begin{tikzpicture}[xscale=0.9]
	\foreach \x in {-1.5, -0.5, 2.5, 3.5}{
	    \draw [fill=black] (\x, 1) circle [radius=0.05];
	    \draw (\x, 1) -- (\x-0.5, 0);
	    \draw (\x, 1) -- (\x+0.5, 0);
	}
	
	\foreach \x in {-2,-1,0,2,3,4}{
		\ifthenelse{\x=-2}
		{\draw [fill=red] (\x, 0) circle [radius=0.1];}
		{\draw [fill=black] (\x, 0) circle [radius=0.1];}
		
	}
    \node[below] at (1,-1) {\phantom{1}};
    \node[above] at (-1.5,1) {\phantom{$\frac{2}{|T|}$}};
	\node at (1,0) {$\cdots$};
	\node at (1,1) {$\cdots$};
	\node[below] at (-2,-0.1) {$s$};
	\end{tikzpicture}
\end{subfigure}
\caption{\emph{Left:} Tightness example for the Ravi-Klein NWST algorithm~\cite{KLEIN1995104} and Algorithm~\ref{alg:greedyPNWST}, with vertex weights indicated and $\OPT=1$. \emph{Right:} Example solution returned by Algorithm~\ref{alg:greedyPNWST}, with weight $2(H_{|T|+1} - 1) \le 2 \ln (|T|+1)$.}
\label{fig:pnwst-tightness}
\end{figure}

\section{Conclusions and future work}
First, by strengthening the analysis of~\cite{Charikar2004}, we showed that \textsc{Priority Steiner tree} is approximable with ratio $\min\{\lceil \log_2 |T| \rceil+1, k\rho\} \le \min\{1.443 \ln |T| + 2, k\rho\}$, and then provided a simple, parallelizable algorithm which achieves the same approximation guarantee. Second, we showed that a natural node-weighted generalization of \textsc{Priority Steiner tree} admits a $O(\log |T|)$-approximation using a generalization of the Ravi-Klein algorithm~\cite{KLEIN1995104} and spider decomposition. 
It remains open whether the approximability gap between $c \log \log n$~\cite{Chuzhoy2008} and $O(\log n)$ for \textsc{Priority Steiner tree} can be tightened, or whether a more efficient approximation algorithm for \textsc{Priority NWST} can be formed. As both problems are a special case of the Steiner tree in directed graphs, this suggests a hierarchy in terms of hardness of approximation.

\paragraph*{Acknowledgments}
The authors wish to thank Alon Efrat and Spencer Krieger for their discussions related to the \textsc{priority NWST} problem.


%
%
%
%

\bibliographystyle{splncs04}
\bibliography{ref}

\end{document}